\documentclass[twocolumn,10pt]{IEEEtran}
\usepackage{braket}

\usepackage{hyperref}
\hypersetup{colorlinks,urlcolor=black}
\usepackage{booktabs}

\input{def1.tex}
\usepackage{dsfont}
\usepackage{minibox}
\DeclareSymbolFont{matha}{OML}{txmi}{m}{it}
\DeclareMathSymbol{\varv}{\mathord}{matha}{118}
\usepackage{eurosym}
\usepackage{hhline,physics}

\usepackage{multicol}
\usepackage{algorithm}
\usepackage{graphicx}
\usepackage{textcomp}
\usepackage{pifont}
\usepackage{float}

\usepackage{subcaption}

\usepackage[multiple]{footmisc}

\usepackage{algorithmic}

\makeatletter
\renewcommand{\baselinestretch}{0.93}

\IEEEoverridecommandlockouts

\begin{document}
	\title{On the Ergodic Capacity  for SIM-Aided Holographic MIMO Communications} 
	\author{Anastasios Papazafeiropoulos, Ioannis Bartsiokas,  Dimitra I. Kaklamani, 			Iakovos S. Venieris \thanks{A. Papazafeiropoulos is with the Communications and Intelligent Systems Research Group, University of Hertfordshire, Hatfield AL10 9AB, U. K. Ioannis Bartsiokas and  Dimitra I. Kaklamani are with the Microwave and Fiber Optics Laboratory, and Iakovos S. Venieris is  with the Intelligent Communications and Broadband Networks Laboratory, School of Electrical and Computer Engineering, National Technical University of Athens, Zografou, 15780 Athens,	Greece.	
			Corresponding author's email: tapapazaf@gmail.com. This work is financially supported by EPSRC/DSIT Federated Telecoms Hub - TITAN (6G-DISCO).}}
	\maketitle\vspace{-1.7cm}
	
	\maketitle
	\begin{abstract}	
		We derive a novel closed-form lower bound on the ergodic capacity of holographic multiple-input multiple-output (HMIMO) systems enhanced by stacked intelligent metasurfaces (SIMs) under Rayleigh fading conditions. The proposed expression is valid for systems with a finite number of antennas and SIM elements and exhibits tightness throughout the whole signal-to-noise ratio (SNR) range. 	Furthermore, we conduct a comprehensive low-SNR analysis, offering meaningful observations  on how key system parameters influence the capacity performance.
	\end{abstract}
	
	
	\section{Introduction}
Considered essential for sixth-generation (6G), reconfigurable intelligent surfaces (RIS) leverage near-passive, controllable units to reconfigure the surrounding propagation environment in real time. \cite{DiRenzo2020,Papazafeiropoulos2021}. A typical RIS comprises a  number of cost-effective elements capable of inducing controllable phase shifts on incident electromagnetic waves, thereby enabling various signal enhancement objectives to support improved wireless connectivity \cite{DiRenzo2020}.
	
	Despite the broad applicability of RIS across diverse communication scenarios due to several inherent advantages, the majority of prior research has primarily focused on single-layer metasurface designs, which inherently limit beam management capabilities \cite{Guo2020a}. Additionally, conventional single-layer RIS designs, constrained by hardware limitations, are generally inadequate for effectively mitigating inter-user interference. 
	
	These constraints spurred the development of the stacked intelligent metasurface (SIM) paradigm  by \textit{An et al.} \cite{An2023}, which offers notable improvements over traditional single-layer RIS. Specifically, a SIM-assisted transceiver architecture was proposed for point-to-point MIMO configurations, with one SIM at the transmitter and another at the receiver, enabling electromagnetic (EM) wave propagation through the metasurfaces. Among the key advantages of SIM technology are its enhanced computational efficiency, ultrafast processing capabilities, lower computational complexity, reduced reliance on RF chains, and decreased energy consumption, which position it as a promising enabler for sixth-generation (6G) communication systems. 	 In this context, the achievable rate performance of SIM-assisted systems in multi-user MISO configurations was examined in \cite{Papazafeiropoulos2024, Papazafeiropoulos2024b} under certain conditions, while aspects of near-field beamforming were explored in \cite{Papazafeiropoulos2024c}.\footnote{\textcolor{black}{It is also important to acknowledge that practical metasurface implementations
			exhibit non-idealities such as phase–amplitude coupling and tuning
			inaccuracies. Recent works on stacked and holographic metasurfaces
			\cite{Li2024,Li2025} have shown that  impairments
			can affect the achievable rate, indicating that hardware-aware SIM modeling
			is an important direction for future research.}}

	In this letter, we first provide a general closed-form lower bound on the ergodic capacity of SIM-assisted MIMO system and a study of the low-SNR regime. Specifically, the proposed general lower bound applies to systems with an \textcolor{black}{arbitrary} number of SIM elements at both ends and maintains tightness across the full SNR regime. \textcolor{black}{	Compared with existing works, this paper advances the analysis of
		SIM-assisted HMIMO systems in several key ways. While  \cite{An2023} and  \cite{Papazafeiropoulos2024}
		investigate achievable rates or hybrid SIM architectures, they do not
		provide closed-form ergodic capacity expressions. Also, classical bounds such
		as \cite{Matthaiou2011a} address correlated MIMO channels but do not model the
		metasurface-induced transformations introduced by  SIMs.  Furthermore, the
		proposed projected-gradient framework enables joint optimization of all
		SIM phase profiles, which is not available in \cite{An2023} or \cite{Matthaiou2011a}.} Moreover, a detailed second-order analysis in the low-SNR region is conducted, yielding closed-form expressions for key performance metrics, which are the minimum energy per bit (EB) and the wideband slope (WS) that characterize MIMO system behavior in this  regime.  

	\section{System Model}\label{System}
	We focus on a point-to-point SIM-aided HMIMO model with SIM integration at both transmission and reception ends. Each SIM interfaces with an intelligent controller responsible for tuning the phase of the electromagnetic wave impinging on each meta-atom of every surface. 
	The transmitter is equipped with $ N_{t} $ antennas, while the receiver is configured with $ N_{r} $ antennas. According to the SIM  architecture  in \cite{An2023,Papazafeiropoulos2024a}, let  $ L $ and $M$ be the total number of layers and the corresponding number of meta-atoms in each layer within the transmit-side SIM. Moreover, let  $ K $ and $N$ be the number of layers and the per-layer meta-atom count at the receiver  SIM. \textcolor{black}{Also, we define  $s\triangleq \min(M,N)$ and $t \triangleq \max(M,N)$}.
	
\subsection{Layer-Wise Phase Control for Stacked Metasurfaces}
Let $\mathcal{L}$ and $\mathcal{K}$ denote the index sets of the transmit- and receive-side layers, respectively. Each transmit layer comprises $M$ meta-atoms indexed by $\mathcal{M}=\{1,\ldots,M\}$, while each receive layer comprises $N$ elements indexed by $\mathcal{N}=\{1,\ldots,N\}$ as mentioned. We adopt unit-modulus phase control at every element. For the $l$-th transmit layer ($l\in\mathcal{L}$), define the per-element phase angles $\theta_{m}^{l}\in[0,2\pi)$, $m\in\mathcal{M}$, and the corresponding complex weights $	\phi_{m}^{l} \triangleq e^{\jmath \theta_{m}^{l}}$, 
which are stacked into the vector and diagonal transmission matrix
\begin{align}
	\boldsymbol{\phi}^{l} \!\triangleq\! \big[\phi_{1}^{l},\ldots,\phi_{M}^{l}\big]^{\T}\!\!\in\mathbb{C}^{M\times 1}, \!\!\quad
	\boldsymbol{\Phi}^{l} \!\triangleq\! \mathrm{diag}\big(\boldsymbol{\phi}^{l}\big)\!\!\in\mathbb{C}^{M\times M}. \label{eq:tx_Phi}
\end{align}
Here, $\boldsymbol{\Phi}^{l}$ captures the element-wise, unit-modulus transmission coefficients applied by the $l$-th layer. Analogously, for the $k$-th receive layer ($k\in\mathcal{K}$), let the per-element phase angles be $\vartheta_{n}^{k}\in[0,2\pi)$, $n\in\mathcal{N}$, with associated complex coefficients $\xi_{n}^{k} \triangleq e^{\jmath \vartheta_{n}^{k}}$, 
assembled as
\begin{align}
	\boldsymbol{\xi}^{k} \!\triangleq\! \big[\xi_{1}^{k},\ldots,\xi_{N}^{k}\big]^{\T}\!\!\in\mathbb{C}^{N\times 1}, 
	\quad \!\!\boldsymbol{\Xi}^{k} \!\triangleq\! \mathrm{diag}\big(\boldsymbol{\xi}^{k}\big)\!\!\in\mathbb{C}^{N\times N}. \label{eq:rx_Xi}
\end{align}
The diagonal matrix $\boldsymbol{\Xi}^{k}$ therefore capures the receive-layer's unit-modulus, element-wise phase profile.

%
	Overall,   the  transmitter and receiver SIMs are expressed as
	\begin{align}
		\bP&=\bPhi^{L}\bW^{L}\cdots\bPhi^{2}\bW^{2}\bPhi^{1}\bW^{1}\in \mathbb{C}^{M \times N_{t}},\label{TransmitterSIM}\\
		\bD&=\bU^{1}\bXi^{1}\bU^{2}\bXi^{2}\cdots\bU^{K}\bXi^{K}\in \mathbb{C}^{N_{r} \times N}.
	\end{align}
		For $l\in\mathcal{L}\setminus\{1\}$, let $\mathbf{W}^{l}\in\mathbb{C}^{M\times M}$ denote the  transmission–coupling matrix that characterizes the interaction between the $(l-1)$-th and $l$-th transmit layers. 
	Similarly, for $k\in\mathcal{K}\setminus\{1\}$, let $\mathbf{U}^{k}\in\mathbb{C}^{N\times N}$ represent the transmission–coupling matrix describing the interaction between the $(k-1)$-th and $k$-th receive layers. 	The external interfaces are modeled as follows. $\mathbf{W}^{1}\in\mathbb{C}^{M\times N_t}$ maps the $N_t$ transmit RF chains (or feed ports) to the first transmit SIM layer, 
	whereas $\mathbf{U}^{1}\in\mathbb{C}^{N_r\times N}$ maps the final receive SIM layer to the $N_r$-element receive antenna array. The  transmission coefficients at both the transmitter and the receiver can be written based on the  Rayleigh-Sommerfeld diffraction theory \cite{An2023}.

	The overall $ \bH \in \mathbb{C}^{N_{r} \times N_{t}} $ channel can be expressed as $	\bH=\bD\bG\bP$, where $
		\bG=\bR^{1/2}_{\mathrm{R}}\tilde{\bG}\bR^{1/2}_{\mathrm{T}}$ 
	is the $N\times M$ HMIMO channel matrix linking the transmit and receive SIM layers  \cite{Hu2022}. Herein,  
	$ \bR_{\mathrm{R}}\in \mathbb{C}^{N\times N} $  and $\bR_{\mathrm{T}}\in \mathbb{C}^{M\times M} $ are the spatial correlation matrices associated with the receive- and transmit-side SIMs, respectively. The mathematical expressions of these  spatial correlation matrices are given by \cite{Demir2022}.
	In addition, $ \tilde{\bG}\sim \mathcal{CN}(\b0,\frac{\beta}{M}\Id_{N}\otimes \Id_{M})\in \mathbb{C}^{N\times M} $ is the   independent  Rayleigh fading channel with   $\beta $ denoting the average path loss between the two  SIMs, which is given by \cite{Rappaport2015}.\footnote{\textcolor{black}{Although $\mathbf{G}$ follows a Kronecker correlation model, the
			composite channel $\mathbf{H}=\mathbf{D}\mathbf{G}\mathbf{P}$ becomes
			highly non-trivial due to SIM-induced multilayer coupling and correlation.
			The derived lower bound incorporates these effects through
			$\mathbf{P}$, $\mathbf{D}$, and the correlation matrices, extending
			classical Wishart-based capacity results to the stacked SIM architecture, which provides the first tractable analytical characterization for such
			systems.}}
	
	The receiver obtains $\by \in \mathbb{C}^{N_{r}}$, which is  given by 
	\begin{align}
		\by=\bH\bs +\bn
	\end{align}
	with 	$\bs \in \mathbb{C}^{N_{t}}$ being the transmitted signal, where  $\EE\{|s_{i}|^{2}\}=\frac{P}{N_t}, i=1,\ldots,N_{t}$ with $P$ expressing the  total transmit power. Also, $\bn\in \mathbb{C}^{N_{r}}\sim \mathcal{CN}(\b0, \sigma^{2}\Id_{N_{r}})$ represents the additive white Gaussian noise (AWGN) with variance $\sigma^{2} $.
	
	\section{Lower Bound on the Ergodic Capacity}\label{ErgodicCapacity}
	By assuming perfect CSI at the receiver and no CSI at the transmitter, uniform power allocation across all data streams is a sensible choice. In this case, the SIM-aided HMIMO  ergodic capacity in bits/s/Hz is expressed as
	\begin{align}
		C_{\mathrm{erg}}=\EE\bigg[\!\log_{2}\!\left(\det\!\left(\Id_{N_{t}}+\frac{\rho}{N_{t}}\bH^{\H}\bH\right)\!\right)\!\bigg],\label{MI}
	\end{align}
	where $\rho=\frac{P}{\sigma^{2}}$ is the signal-to-noise ratio (SNR).
	\begin{theorem}\label{th1}
	The ergodic capacity of a  SIM-aided HMIMO system admits a lower bound given by
		\begin{align}
			C_{\mathrm{erg}}&\ge 	C_{\mathrm{LB}}=N_{t} \log_{2}\bigg(\!\!1+\frac{\rho}{N_{t}} \exp \Big(\frac{1}{N_{t}}\big(\textcolor{black}{\sum_{i=0}^{s-1}\psi(t-i)}\nn\\
			&+\ln |\bP\bP^{\H}|+\ln |\bD^{\H}\bD|+\ln |\bR_{\mathrm{T}}|+|\bR_{\mathrm{R}}|\big)\!\Big)\!\bigg)\label{capacity0}
		\end{align}
		with $\psi(x)$ being the digamma function [13, eq. (8.360.1)].
	\end{theorem}
	\begin{proof}
		Please see Appendix~\ref{th1proof}.	
	\end{proof}
	\begin{remark}
		We observe that $	C_{\mathrm{LB}} $ simplifies to the lower bound obtained in \cite[Eq. 5]{Matthaiou2011a}.
	\end{remark}
	
	The optimization problem is mathematically described as 
	\begin{subequations}\label{eq:subeqns}
		\begin{align}
			(\mathcal{P})~~&\max_{\bphi_{l},\bxi_{k}} 	\;C_{\mathrm{LB}}\label{Maximization1} =f(\bphi_{l},\bxi_{k}),\\
			&	\bP=\bPhi^{L}\bW^{L}\cdots\bPhi^{2}\bW^{2}\bPhi^{1}\bW^{1},
			\label{Maximization3} \\
			&	\bZ=\bU^{1}\bXi^{1}\bU^{2}\bXi^{2}\cdots\bU^{K}\bXi^{K},
			\label{Maximization4} \\
			&		\bPhi^{l}=\diag(\phi^{l}_{1}, \dots, \phi^{l}_{M})\quad \!	\bXi^{k}=\diag(\xi^{k}_{1}, \dots, \xi^{k}_{N}),
			\label{Maximization6} \\
			&	|	\phi^{l}_{m}|=1, 	\!\!\!\!\!\!\!\!\;\;\;\;\;\!\!~\!		|\xi^{k}_{n}|=1	\label{Maximization8}.
		\end{align}
	\end{subequations}
	
	It is evident that the optimization problem $ (\mathcal{P}) $ is classified as nonconvex, since the objective function does not exhibit concavity or convexity with respect to the optimization variables and is further constrained by non-convex constant modulus conditions. Prior work on SIM-assisted systems has predominantly utilized alternating optimization (AO), where the transmit and receive SIM phase shifts are updated alternately. Although AO is relatively straightforward to implement, its convergence often requires numerous iterations, particularly as the size of the SIM configuration  increases \cite{Perovic2021}. In the context of SIM-assisted HMIMO architectures, where every individual metasurface  typically comprises a substantial number of elements, AO becomes computationally inefficient. These considerations encourage the adoption of a more effective approach based on the projected gradient method, inspired by techniques in \cite{Perovic2021}, which enables concurrent optimization of the phase shifts at both the transmitter- and receiver-side SIMs. 	
	
	To address the optimization problem in \eqref{eq:subeqns}, we introduce the proposed algorithm in Algorithm \ref{Algoa1}. The core principle involves initializing from a feasible point 
	$ (\bphi_{l}^{0},\bxi_{k}^{0}) $ and iteratively updating the variables along the gradient direction	 $ \nabla f(\bphi_{l},\bxi_{k}) $. The step size for each update is governed by the positive parameters $\mu_{n}^{q} >0$, for  $ q=1,2 $.
	
	The algorithm operates within the subsequent feasible sets
	\begin{align}
		\Phi_{l}&=\{\bphi_{l}\in \mathbb{C}^{M \times 1}: |\phi^{l}{i}|=1, ; i=1,\ldots, M\},\\
		\Xi_{k}&=\{\bxi_{k}\in \mathbb{C}^{N \times 1}: |\xi^{k}_{i}|=1, ; i=1,\ldots, N\}.
	\end{align}
	
	At each iteration, before proceeding along the gradient direction of 	$ f(\bphi_{l},\bxi_{k}) $, the updated variables are projected back onto their respective feasible sets 
	$ \Phi_{l} $, and $ \Xi_{k} $. This ensures that the iterates remain within the allowable domain of the problem. For completeness, we also provide the gradient components 
	$\nabla_{\bphi_{l}}f(\bphi_{l},\bxi_{k}) $ and $ \nabla_{\bxi_{k}}f(\bphi_{l},\bxi_{k}) $, which represent the directions of the steepest ascent for the objective function $f$, as defined in \cite[Theorem 3.4]{hjorungnes:2011}. 	The operators 	$ P_{\Phi_{l}}(\cdot) $ and $ P_{\Xi_{k}}(\cdot) $ correspond to projections onto 
	$ \Phi_{l} $ and $ \Xi_{k} $, respectively.
	
	\begin{algorithm}[th]
		\caption{Projected Gradient Ascent Method \label{Algoa1}}
		\begin{algorithmic}[1]
			\STATE Input: $\bphi_{l}^{0},\bxi_{k}^{0},\mu_{i}^{\textcolor{black}{q}}>0$ \textcolor{black}{for $ q=1,2 $}.
			\STATE \textbf{for} $ i=1,2,\ldots \textbf{do} $
			\STATE ~~~~~$\bphi_{l}^{i+1}=P_{\Phi_{l}}(\bphi_{l}^{i}+\mu_{i}^{\textcolor{black}{1}}\nabla_{\bphi_{l}}f(\bphi_{l}^{i},\bxi_{k}^{i}))$
			\STATE ~~~~~$\bxi_{k}^{i+1}=P_{\Xi_{k}}(\bxi_{k}^{i}+\mu_{i}^{\textcolor{black}{2}}\nabla_{\bxi_{k}}f(\bphi_{l}^{i},\bxi_{k}^{i}))$
			\STATE \textbf{end for}
		\end{algorithmic}
	\end{algorithm}

	\begin{lemma}\label{lemmaGradient}
		The complex derivatives of 	$f(\bphi_{l},\bxi_{k}) $
		regarding the conjugate variables 
		$\bphi_{l}^{*}$ and $\bxi_{k}^{*}$	are computed as  
		\begin{align}
			\nabla_{\bphi_{l}}f(\bphi_{l},\bxi_{k})&=\bar{V}\diag(\bA_{l}(\bP^{\H}\bP)^{-1} ) ,\label{gradient2}\\
			\nabla_{\bxi_{k}}f(\bphi_{l},\bxi_{k})&=\bar{V}\diag(\bC_{k}(\bD\bD^{\H})^{-1} ) ,\label{gradient3}
		\end{align}
		where $\bar{V}=\frac{\rho V}{  \ln2 N_{t}  \bigg(\!\!1+\frac{\rho}{N_{t}} V\!\bigg)}$, 
		\begin{align}
			\bA_{l}&\!=\!\bW^{l}\bPhi^{l-1}\bW^{l-1}\!\cdots\! \bPhi^{1}\bW^{1}\bPhi^{L}\bW^{L}\cdots\bPhi^{l+1}\bW^{l+1},\\
				\bC_{k}&\!=\!\bU^{l}\bXi^{l-1}\bU^{l-1}\!\cdots\! \bXi^{1}\bU^{1}\bXi^{L}\bU^{L}\cdots\bXi^{l+1}\bU^{l+1}.
		\end{align}
	\end{lemma}
	\begin{proof}
		Please see Appendix~\ref{lem1}.	
	\end{proof}
	
	The  optimization framework of the SIM leverages the gradient ascent method, offering a notable benefit due to the availability of a closed-form expression for the gradient. \textcolor{black}{Since the updates constitute projected gradient
	ascent over compact unit-modulus sets and $C_{\mathrm{LB}}$ is smooth and
	bounded, standard results guarantee convergence to a stationary point for
	suitable step sizes \cite{An2023,Papazafeiropoulos2024}. The convergence curve in the numerical section
	further confirms rapid monotonic convergence in practice}.
	 
\textcolor{black}{	The computational complexity of Algorithm~1 is mainly determined by the evaluation of the lower-bound expression and its gradients in~\eqref{capacity0}, \eqref{gradient2}, and \eqref{gradient3}. Each iteration requires computing the composite SIM matrices on the transmit and receive sides, which incurs a cost of $\mathcal{O}(L M^{2}N_{t} + K N^{2}N_{r})$. The subsequent formation and inversion of the covariance-related matrices, as well as the determinant and gradient calculations, contribute an additional $\mathcal{O}(M^{3} + N^{3} + N_{t}^{3} + N_{r}^{3})$. Therefore, the overall per-iteration complexity is $\mathcal{O}(L M^{2}N_{t} + K N^{2}N_{r} + M^{3} + N^{3} + N_{t}^{3} + N_{r}^{3})$, which corresponds to the cost of computing all terms required by the gradient-based updates in Algorithm~1.} \textcolor{black}{More specifically, this cost scales linearly with the number of SIM layers, while
		scaling cubic in the meta-atom and antenna dimensions.}


	\section{Low-SNR Regime}
	According to \cite{Lozano2003}, the operation of MIMO systems in the low-SNR region is better captured by the normalized transmit EB $E_{b}/N_{0}$ as opposed to the conventional per-symbol SNR metric. This leads to the following capacity formulation $	C_{\mathrm{erg}}\left(\frac{E_{b}}{N_{0}}\right)\approx S_{0} \log_{2}\left(\frac{\frac{E_{b}}{N_{0}}}{\frac{E_{b}}{{N_{0}}_{\min}}}\right)$, 
	where $\frac{E_{b}}{{N_{0}}_{\min}}$ denotes the minimum normalized EB necessary to support any positive transmission rate reliably, while $S_{0}$ represents the WS, both serving as key indicators of low-SNR performance. Based on \cite{Verdu2002}, these two key performance metrics are defined as $\frac{E_{b}}{{N_{0}}_{\min}}=\frac{1}{\dot	C_{\mathrm{erg}}(0)}, ~S_{0}=-2\ln2\frac{\left(\dot C_{\mathrm{erg}}(0)\right)^{2}}{\ddot C_{\mathrm{erg}}(0)}$,  
	where $\dot{C}_{\mathrm{erg}}(0)$ and $\ddot{C}_{\mathrm{erg}}(0)$ represent the first- and second-order derivatives of the ergodic capacity in \eqref{MI} with respect to the SNR $\rho$, respectively. As outlined in \cite{Lozano2003}, analyzing the low-SNR regime involves utilizing the concept of the dispersion of a random matrix.
	\begin{definition}
		The dispersion of an $s \times s$ random matrix is defined as follows $		\zeta(\bA)=s\frac{\EE[\tr(\bA^{2})]}{\tr(\bA)}.$
	\end{definition}
	
	\begin{theorem}\label{th2}
		In a SIM-assisted HMIMO system with dimensions $N_{r} \times N_{t}$, the minimum EB and the WS are given by
		\begin{align}
			\frac{E_{b}}{{N_{0}}_{\min}}&=N_{t}\ln 2(\mathrm{tr}\left(\bR_{\mathrm{T}}\bP\bP^{\H}\right)\mathrm{tr}\left(\bR_{\mathrm{R}}\bD^{\H}\bD\right))^{-1},\label{minEnergy}\\
			S_{0}&=\frac{2 N_{t}N_{r}}{N_{t}\zeta(\bR_{\mathrm{T}}\bP\bP^{\H})+N_{r}\zeta(\bR_{\mathrm{R}}\bD^{\H}\bD)}.\label{slope}
		\end{align}
	\end{theorem}
	\begin{proof}
		Please see Appendix~\ref{th2proof}.	
	\end{proof}
	
	It is worth highlighting that, in SIM-assisted HMIMO systems, the minimum EB is influenced by the spatial correlation, whereas in conventional MIMO systems, it remains unaffected by such correlation. Furthermore, increasing the number of transmit antennas $ N_{t}$ generally leads to a higher minimum EB. Notably, both expressions depend on the SIMs in terms of the matrices $\bP$ and $\bD$. Lastly, under the assumptions of i.i.d. Rayleigh fading and no SIMs, expressions \eqref{minEnergy} and \eqref{slope} reduce to 
	the results presented in \cite[eq. (17)]{Lozano2003} and \cite[eq. (19)]{Lozano2003}, respectively.
	
	Similar to the optimization problem $ (\mathcal{P}) $, two problems can be formulated with the same constraints, one for the  minimum EB and one for the WS. The procedure is omitted since it is the same as the methodology followed in Section \ref{ErgodicCapacity}.
	
	The corresponding gradients are obtained as
	\begin{align}
		&\nabla_{\bphi_{l}}	\frac{E_{b}}{{N_{0}}_{\min}}=-N_{t}\ln 2\frac{\diag(\bA_{l}\bR_{\mathrm{T}}\bP^{\H}\bP )}{(\mathrm{tr}\left(\bR_{\mathrm{T}}\bP\bP^{\H}\right)\mathrm{tr}\left(\bR_{\mathrm{R}}\bD^{\H}\bD\right))^{2}},\\
		&	\nabla_{\bxi_{k}}	\frac{E_{b}}{{N_{0}}_{\min}}=-N_{t}\ln 2\frac{\diag(\bC_{k}\bR_{\mathrm{R}}\bD\bD^{\H} )}{(\mathrm{tr}\left(\bR_{\mathrm{T}}\bP\bP^{\H}\right)\mathrm{tr}\left(\bR_{\mathrm{R}}\bD^{\H}\bD\right))^{2}}, \\
		&	\nabla_{\bphi_{l}}	S_{0}=-\frac{2 N_{t}^{2}N_{r} s}{(N_{t} \zeta(\bR_{\mathrm{T}}\bP\bP^{\H})+N_{r}\zeta(\bR_{\mathrm{R}}\bD^{\H}\bD))^{2}\tr^{2}(\bR_{\mathrm{T}}\bP\bP^{\H})}\\
		&\times \!\diag\!\left(\!\right.2 \tr\left(\right.\!\bR_{\mathrm{T}}\bP\bP^{\H})\bA_{l}\bP^{\H}\bP\bR_{\mathrm{T}}\bP^{\H}\bR_{\mathrm{T}} \\
		&-\tr\! \left(\right.\!\!(\bR_{\mathrm{T}}\bP\bP^{\H}\!\left. \right)^{2}\bA_{l}\bR_{\mathrm{T}}\!\left. \right)\!\!\left. \right),\\
		&\nabla_{\bxi_{k}}S_{0}=-\frac{2 N_{t}^{2}N_{r} s}{(N_{t} \zeta(\bR_{\mathrm{T}}\bP\bP^{\H})+N_{r}\zeta(\bR_{\mathrm{R}}\bD^{\H}\bD))^{2}\tr^{2}(\bR_{\mathrm{R}}\bD^{\H}\bD)}\\
		&\times \!\diag\!\left(\!\right.2 \tr\left(\right.\!\bR_{\mathrm{R}}\bD^{\H}\bD)\bC_{k}\bD\bD^{\H}\bR_{\mathrm{R}}\bD\bR_{\mathrm{R}} \\
		&-\tr\! \left(\right.\!\!(\bR_{\mathrm{R}}\bD^{\H}\bD\!\left. \right)^{2}\bC_{k}\bR_{\mathrm{R}}\!\left. \right)\!\!\left. \right).
	\end{align}
	\begin{proof}
		The proof is omitted due to limited space and because it follows the steps of the proof of Lemma~\ref{lemmaGradient} since the dependence of $	\frac{E_{b}}{{N_{0}}_{\min}}$ and $S_{0}$ from $\bP$ and $\bD$ is similar to the dependence  of $C_{\mathrm{LB}}$ from these parameters.	
	\end{proof}
	
	The complexity of the two algorithms is similar to  problem $(\mathcal{P})$ since their implementation involves similar  matrix operations.
	\section{Numerical Results}\label{Numerical}
	The analytical expressions and Monte Carlo (MC) simulation results are used for validation. In the simulated setup, the SIMs at the transmitter and receiver are placed along  the $x$–$z$ plane, with their centers being along the 
	$y$-axis at an altitude of of $H_{\mathrm{BS}} = 5~\mathrm{m}$. The design and configuration of both SIMs adhere to the methodology outlined in \cite{An2023,Papazafeiropoulos2024a}. 	Each meta-atom occupies an area of $(\lambda/2)^2$, and the distance separating neighboring meta-atoms is fixed at $\lambda/2$. The gap between the two metasurfaces is given by $d_{\mathrm{SIM}} = T_{\mathrm{SIM}}/L$, where $T_{\mathrm{SIM}} = R_{\mathrm{SIM}} = 5\lambda$ denotes the SIM thickness. A path-loss exponent of $b = 2.5$ is adopted. The separation between the transmitter- and receiver-side SIMs is fixed at $d=200~\mathrm{m}$. 
	The system employs a bandwidth of $20~\mathrm{MHz}$ at a carrier frequency of $2~\mathrm{GHz}$. 
	The total transmit power is set to $\rho=20~\mathrm{dBm}$, and the receiver noise (sensitivity) level is taken as $\sigma^{2}=-110~\mathrm{dBm}$. 
	\textcolor{black}{By default, the antenna arrays comprise $N_t=N_r=8$ elements, each SIM contains $M=40, N=100$ meta-atoms, and the number of SIM layers is $K=L=4$} \textcolor{black}{\cite{An2023,Papazafeiropoulos2024a}}.
	
	 \begin{figure*}[t]
	 	\begin{minipage}{0.33\textwidth}
	 		\centering
	 		\captionbox{Analytical lower bound and simulated ergodic capacity  versus the SNR.\label{fig1}}{	\includegraphics[trim=0cm -1.2cm 0cm 0.2cm, clip=true, width=2.2in]{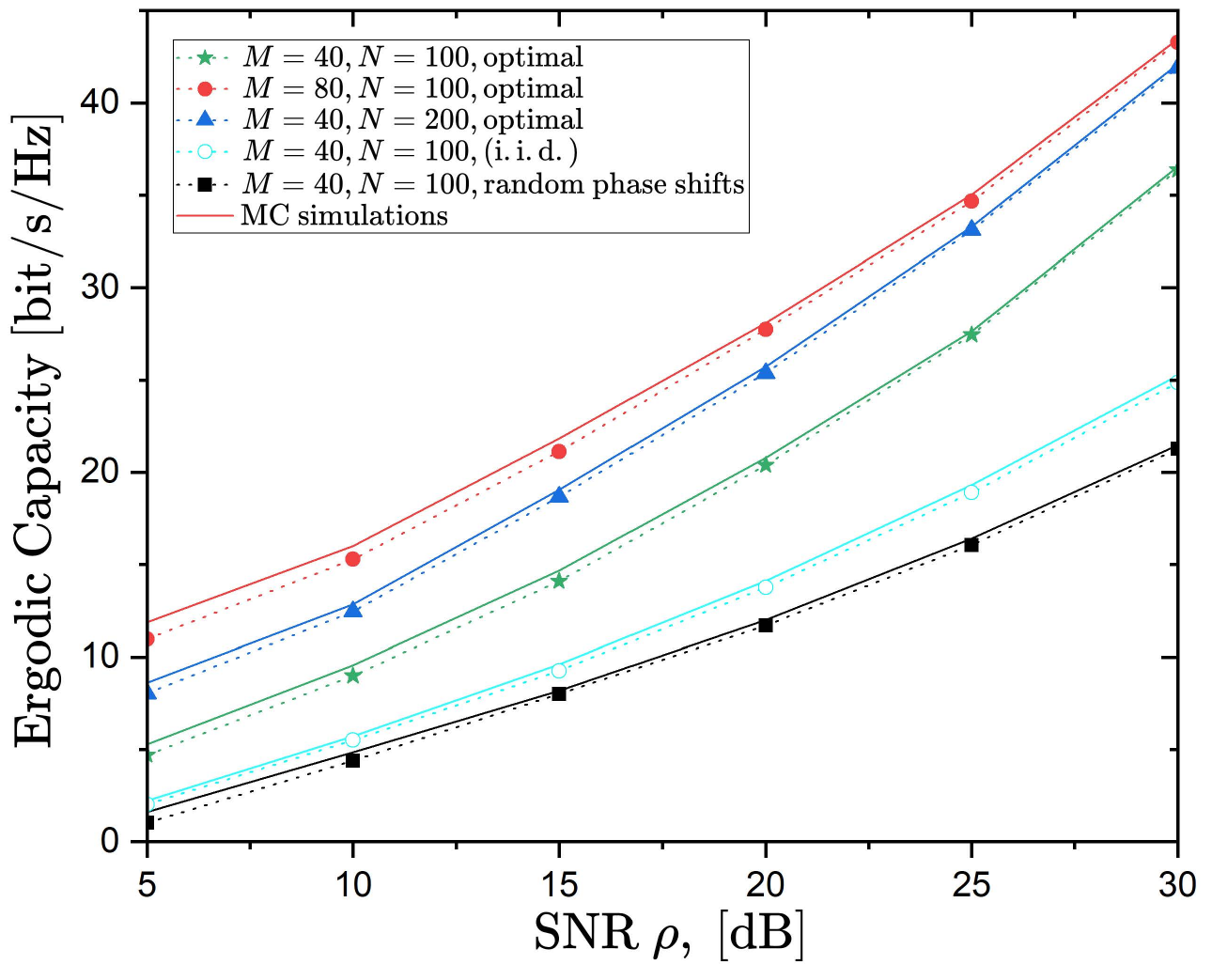}}
	 		\vspace*{-0.2cm}
	 	\end{minipage}
	 	\begin{minipage}{0.33\textwidth}
	 		\centering
	 		\captionbox{Simulated and analytical low-SNR ergodic capacity versus the transmit EB.\label{fig2}}{	\includegraphics[trim=0cm 0.0cm 0cm 0.0cm, clip=true, width=2.2in]{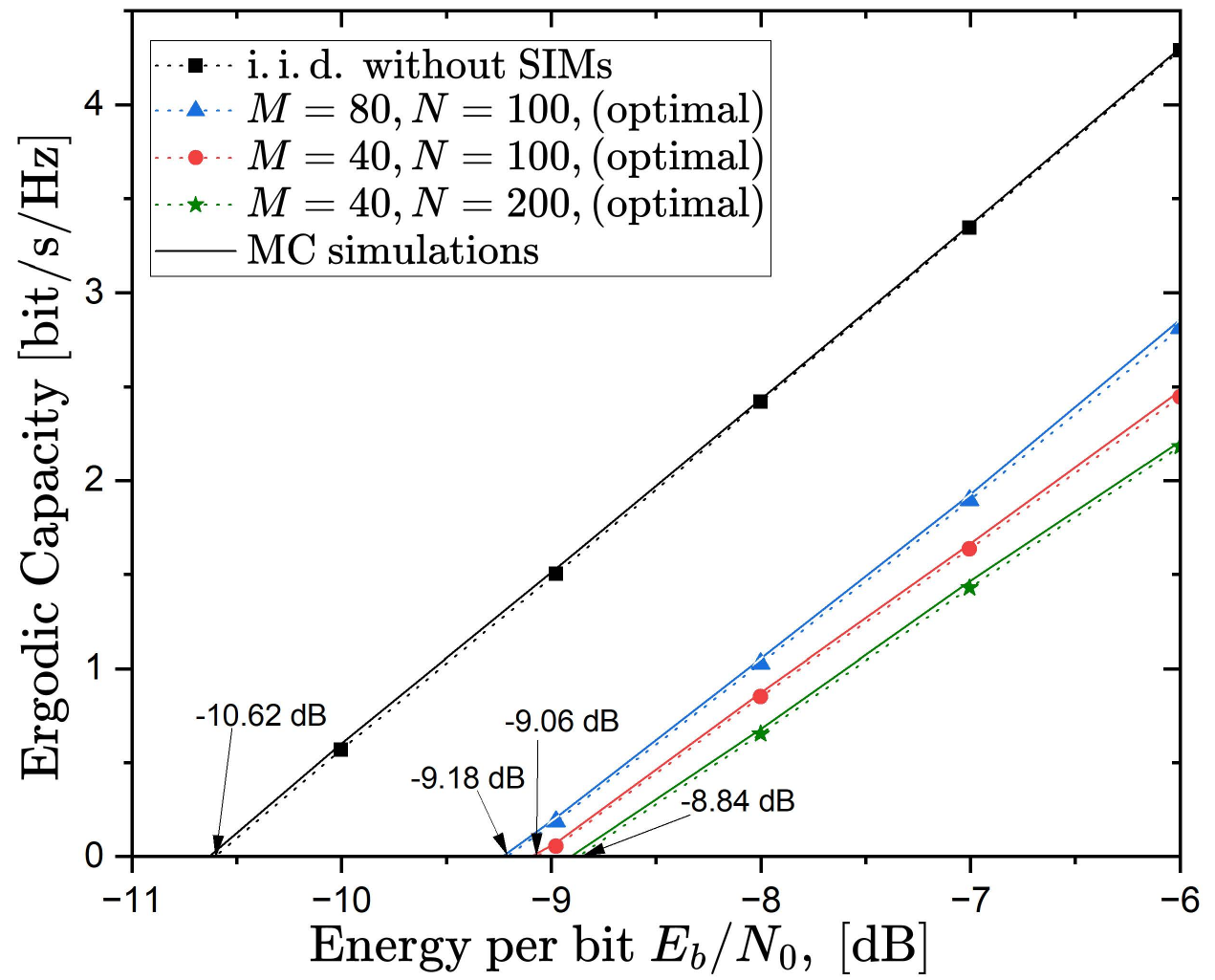}}
	 		\\ 
	 		\vspace*{-0.2cm}
	 	\end{minipage}
	 	\begin{minipage}{0.33\textwidth}
	 	\centering
	 	\captionbox{Analytical lower bound with respect to the number of iterations.\label{fig3}}{	\includegraphics[trim=0cm -1.80cm 0cm 0.2cm, clip=true, height=1.9in, width=2.1in]{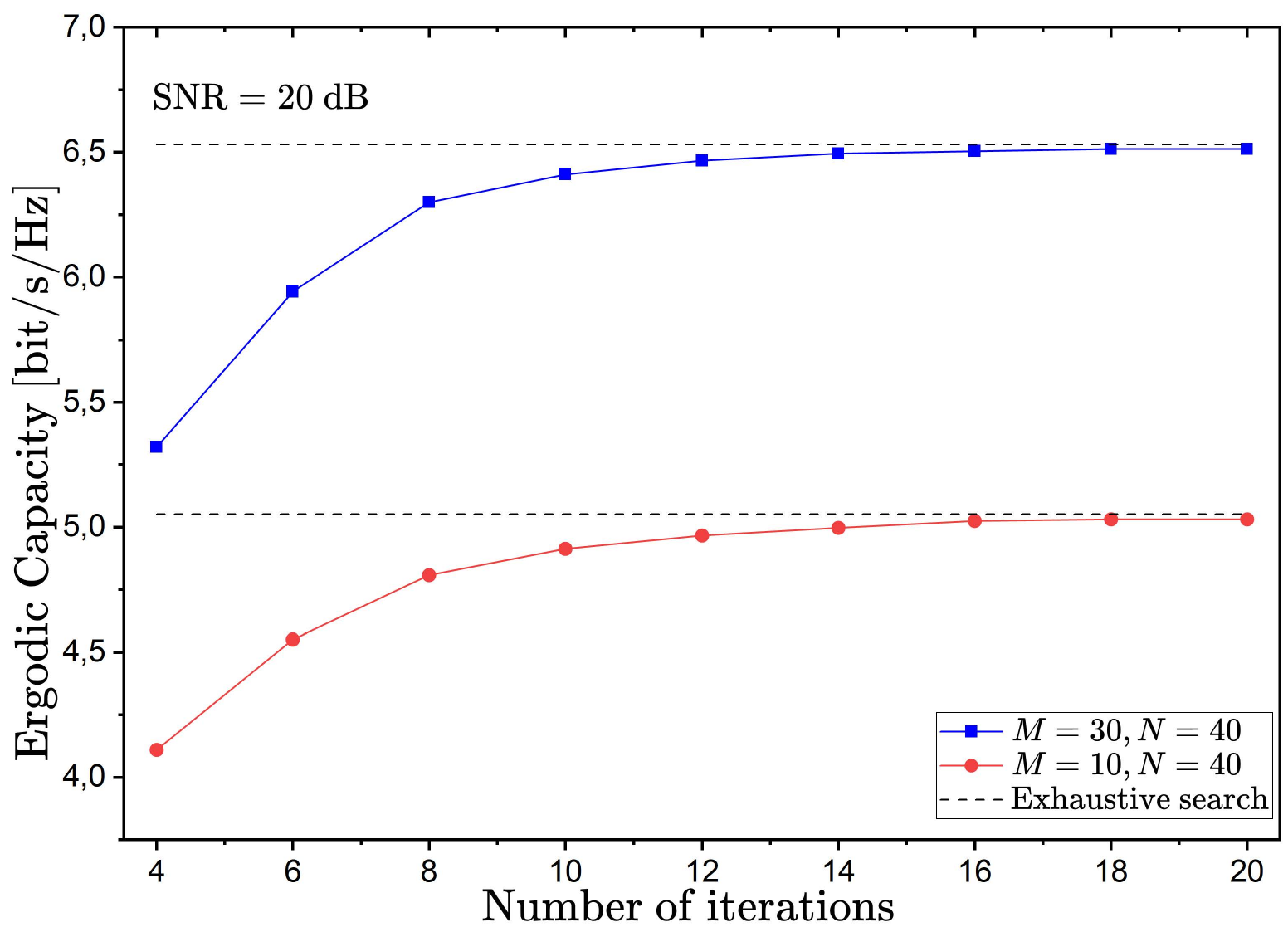}}
	 	\\ 
	 	\vspace*{-0.2cm}
	 \end{minipage}
	 	\label{Fig22}
	 \end{figure*}
	 
	\textcolor{black}{In Fig. \ref{fig1}, we  observe an excellent match
	between the analytical lower bound and Monte Carlo simulations for
	various $(M,N)$ configurations, confirming the tightness of the bound
	across the SNR range. As expected, increasing the number of SIM
	meta-atoms improves the ergodic capacity. However, the figure reveals an
	important asymmetry: enlarging $s=\min(M,N)$ yields a much stronger
	capacity gain than enlarging $t=\max(M,N)$, which  is fully
	consistent with the analytical expression. The figure further includes the i.i.d. channel case and the
	scenario with unoptimized (random) SIM phase shifts, both showing 
	lower performance and underscoring the impact of phase optimization on
	the overall SIM-assisted HMIMO gain.}
	
	Fig. \ref{fig2} shows the simulated and analytical low-SNR capacity versus the transmit EB ${E_{b}}/{N_{0}}$ according to Theorem \ref{th2}. For comparative purposes, we have depicted the scenario of a conventional i.i.d. Rayleigh MIMO channel. \textcolor{black}{Fig.~\ref{fig2} highlights a notable trade-off: although SIMs improve the
		ergodic capacity at moderate and high SNRs, they increase the minimum
		energy per bit $E_b/N_0^{\min}$. From Theorem~2, this metric scales with
		$(\mathrm{tr}(\mathbf{R}_T\mathbf{P}\mathbf{P}^H)\,
		\mathrm{tr}(\mathbf{R}_R\mathbf{D}\mathbf{D}^H))^{-1}$, which becomes
		larger when the SIM induces stronger spatial correlation. At very low SNR, this correlation dominates the
		channel behavior and reduces the effective rank, requiring more energy
		to reliably transmit even vanishingly small rates. This also explains
		the slight reduction in the wideband slope $S_0$ in Fig.~\ref{fig2}. In summary,
		SIMs enhance capacity when power is sufficient, but in the energy-limited
		regime their correlation-enhancing effect leads to higher required
		$E_b/N_0^{\min}$.}
		
		\textcolor{black}{			Fig.~\ref{fig3} illustrates the convergence  of the projected gradient ascent scheme. In particular, $C_{\mathrm{LB}}$ increases monotonically and stabilizes within few
			iterations for different initializations. The algorithm finishes  when the iterations are greater than $20$ or the difference  between the two last iterations is less than $10^{-5}$. A small-scale comparison with
			exhaustive search further indicates that any performance loss due to the
			projection step is negligible in practice. }
%

\textcolor{black}{	\noindent\textbf{Discussion on the Impact of SIM Parameters:}
	The trends observed in Fig.~\ref{fig1} and Fig.~\ref{fig2} illustrate the
	impact of the SIM aperture dimensions $(M,N)$. Increasing the number of
	meta-atoms per layer enlarges the effective SIM aperture and enables finer
	sampling of the impinging wavefield, which enhances the spatial
	degrees of freedom captured by the composite matrices $\mathbf{P}$ and
	$\mathbf{D}$. This leads to a more favorable eigenvalue distribution of
	$\mathbf{G}^H\mathbf{G}$ and results in the monotonic capacity gains
	visible in Fig.~\ref{fig1}. The figure also shows that changes in
	$s=\min(M,N)$ have a more pronounced effect than changes in
	$t=\max(M,N)$, which means that the smallest SIM play a key role in the performance. Furthermore, increasing the number of layers $L$ and $K$ (not shown due to limited space) introduces
	additional cascaded diffraction, which enriches the
	spatial degrees of freedom embedded in the composite matrices $\mathbf{P}$
	and $\mathbf{D}$, i.e., the capacity becomes larger. These effects extend to the low-SNR regime in Fig.~\ref{fig2}. The metrics
	$E_b/N_0^{\min}$ and $S_0$ depend on
	$\mathrm{tr}(\mathbf{R}_T\mathbf{P}\mathbf{P}^H)$ and
	$\mathrm{tr}(\mathbf{R}_R\mathbf{D}\mathbf{D}^H)$, both of which increase
	with the SIM aperture size.}

	\section{Conclusion} \label{Conclusion} 
	This letter presented a lower bound formulation for the ergodic capacity of SIM-assisted HMIMO systems. The resulting expression is mathematically tractable, valid for a finite number of antennas and SIM elements, and exhibits  tightness across the SNR. In addition, a second-order capacity expansion was performed in the low-SNR regime, leading to closed-form expressions for key performance indicators, namely the minimum EB and the WS, offering valuable insights into the    performance of the system. \textcolor{black}{Relevant extensions include accounting for hardware impairments, imperfect
		CSI, and practical phase quantization in SIM layers. Incorporating these
		constraints is essential for translating SIM-assisted HMIMO into real-world
		deployments.}
	\begin{appendices}
		\section{Proof of Theorem~\ref{th1}}\label{th1proof}	
	\textcolor{black}{	Starting from (12),
			Minkowski's determinant inequality yields
		\[
		\det(\mathbf{I}_{N_t}+\tfrac{\rho}{N_t}\mathbf{H}\mathbf{H}^H)^{1/N_t}
		\ge 1+\tfrac{\rho}{N_t}\det(\mathbf{H}\mathbf{H}^H)^{1/N_t},
		\]
		leading to
		\[
		C_{\mathrm{erg}}
		\ge N_t\,\mathbb{E}\!\left[
		\log_2\!\left(1+\tfrac{\rho}{N_t}
		e^{\,\frac{1}{N_t}\ln\det(\mathbf{H}\mathbf{H}^H)}\right)
		\right].
		\]
		Since $f(x)=\log_2(1+a e^x)$ is convex for $a>0$, Jensen's inequality gives
		\[
		C_{\mathrm{erg}}
		\ge N_t\log_2\!\left(
		1+\tfrac{\rho}{N_t}
		e^{\,\frac{1}{N_t}\mathbb{E}[\ln\det(\mathbf{H}\mathbf{H}^H)]}
		\right).
		\]}
		Using $\mathbf{H}=\mathbf{D}\mathbf{G}\mathbf{P}$ and
		$\det(\mathbf{AB})=\det(\mathbf{A})\det(\mathbf{B})$, we have
		\begin{align}
			\ln\det(\mathbf{H}\mathbf{H}^H)
		&= \ln\det(\mathbf{G}^H\mathbf{G})
		+ \ln|\mathbf{P}\mathbf{P}^H|\nn\\
		&		+ \ln|\mathbf{D}^H\mathbf{D}|
		+ \ln|\mathbf{R}_T| + \ln|\mathbf{R}_R|.\label{eq1}
	\end{align}
	
	\textcolor{black}{For the Kronecker MIMO model, $\mathbf{G}\in\mathbb{C}^{N\times M}$
	yields a Gram matrix with $\operatorname{rank}(\mathbf{G}^H\mathbf{G})=
	\min(M,N)$. Using the standard Wishart result \cite[A.8.1]{Grant2002}, we
	obtain
	\[
	\mathbb{E}\!\left[\ln\det(\mathbf{G}^H\mathbf{G})\right]
	= \sum_{i=0}^{\min(M,N)-1}\psi\big(\max(M,N)-i\big),
	\]
	which, when substituted into~\eqref{eq1}, yields Theorem~1.}

		\section{Proof of Lemma~\ref{lemmaGradient}}\label{lem1}
		The step-by-step computation of  	$ \nabla_{\bphi_{l}}f(\bphi_{l},\bxi_{k}) $ starts with computing its differential, given by
		\begin{align}
			&	d(f(\bphi_{l},\bxi_{k}))
			\!=\! \frac{\rho V d(|\bP\bP^{\H}|)}{ \ln2 N_{t} \bigg(\!\!1+\frac{\rho}{N_{t}} V\!\bigg)}\label{differentialPhi2}\\
			&\!=\! \frac{ \rho V\tr((\bP\bP^{\H})^{-1} (d(\bP)\bP^{\H}+\bP d(\bP^{\H})  )) }{ \ln2  N_{t} \bigg(\!\!1+\frac{\rho}{N_{t}} V\!\bigg)},\label{differentialPhi3}
		\end{align}
		where $ V=\exp \Big(\textcolor{black}{\frac{1}{N_{t}}}\big(\textcolor{black}{\sum_{i=0}^{s-1}\psi(t-i)}+\ln |\bP\bP^{\H}|+\ln |\bD^{\H}\bD|+\ln |\bR_{\mathrm{T}}|+|\bR_{\mathrm{R}}|\big)\!\Big)$. 	In \eqref{differentialPhi3}, we have made use of the property $d(|\bX|)=|\bX|\tr(\bX^{-1} d \bX)$ for a square, invertible matrix $\bX \in \mathbb{C}^{n \times n}$.

		Additionally, the expression in \eqref{TransmitterSIM} admits the following differential form\begin{align}
			d(\bP)&=\bPhi^{L}\bW^{L}\cdots\bPhi^{l+1}\bW^{l+1}d(\bPhi^{l})\bW^{l}\bPhi^{l-1}\nn\\
			&\times\bW^{l-1}\cdots\bPhi^{1}\bW^{1}. \label{differentialPhi4}
		\end{align}
		Upon replacing $d(\bP)$ with its expression in \eqref{differentialPhi3}, we derive
		\begin{align}
			&	d(f(\bphi_{l},\bxi_{k}))\nn\\
			&		\!=\! \frac{\rho V \tr((\bP\bP^{\H})^{-1} (\bA_{l} d(\bPhi^{l}) +\bA_{l}^{\H} d((\bPhi^{l})^{\H})  )) }{ \ln2  N_{t} \bigg(\!\!1+\frac{\rho}{N_{t}} V\!\bigg)},\label{differentialPhi5}
		\end{align}
		where 
		\begin{align}
			\bA_{l}&\!=\!\bW^{l}\bPhi^{l-1}\bW^{l-1}\!\cdots\! \bPhi^{1}\bW^{1}\bPhi^{L}\bW^{L}\cdots\bPhi^{l+1}\bW^{l+1}.
		\end{align}

		The gradient is obtained as
		\begin{align}
			\nabla_{\bphi_{l}}f(\bphi_{l},\bxi_{k})&=	\frac{\rho V\diag(\bA_{l}(\bP^{\H}\bP)^{-1} ) }{ \ln2 N_{t} \bigg(\!\!1+\frac{\rho}{N_{t}} V\!\bigg)}.
		\end{align}
		In a similar way, $\nabla_{\bxi_{k}}f(\bphi_{l},\bxi_{k})$ is obtained.
		
		\section{Proof of Theorem~\ref{th2}}\label{th2proof}
		The derivation commences by recognizing that $\frac{\mathrm{d}}{\mathrm{d} x}	\ln(\det(\Id+x \bA))\Bigr|_{\substack{x=0}}=\tr (\bA)$, 
		which leads to the need to compute $\mathbb{E}\big[\mathrm{tr}\left(\mathbf{H}^{\mathrm{H}}\mathbf{H}\right)\big]$. Nonetheless, performing this computation directly can be quite cumbersome. To address this challenge, we employ tools from random matrix theory, drawing upon the methodology presented in \cite{Lozano2003}. Specifically, by simultaneously considering the impact of spatial correlations and passive beamforming matrices, the framework outlined in \cite[Appendix B]{Lozano2003}  becomes applicable. Under these conditions, the prerequisites of \cite[Lemma 3]{Lozano2003} are satisfied, allowing the expression to be simplified as follows
		\begin{align}
			\mathbb{E}\big[\mathrm{tr}\left(\mathbf{H}^{\mathrm{H}}\mathbf{H}\right)\big]
			&=\mathrm{tr}\left(\bR_{\mathrm{T}}\bP\bP^{\H}\right)\mathrm{tr}\left(\bR_{\mathrm{R}}\bD^{\H}\bD\right)\label{MI3}.
		\end{align}
		
		From the combination of \eqref{MI} and \eqref{MI3}, equation \eqref{minEnergy} follows directly.
		The WS can be derived using the formulation provided in \cite[eq. (19)]{Lozano2003} as
		\begin{align}
			S_{0}&=\frac{2 N_{t}N_{r}}{N_{t}\zeta(\bR_{\mathrm{T}}\bP\bP^{\H})+N_{r}\zeta(\bR_{\mathrm{R}}\bD^{\H}\bD)}.
		\end{align}

	\end{appendices}

	\bibliographystyle{IEEEtran}

	\bibliography{IEEEabrv,bibl}
\end{document}